\documentclass[letterpaper, 10 pt, conference]{IEEEconf}  
\usepackage[utf8]{inputenc}

\IEEEoverridecommandlockouts    
\usepackage{cite}
\usepackage{amsmath,amssymb,amsfonts}
\usepackage{graphicx}
\usepackage{bm}
\usepackage{multirow}

\usepackage{algorithm}
\usepackage{algorithmicx}
\usepackage{algpseudocode}

\usepackage{booktabs}

\usepackage{textcomp}
\usepackage{amssymb}
\usepackage{amsmath}
\usepackage{graphicx}
\usepackage{xcolor}
\usepackage{makecell}
\usepackage{comment}
\usepackage{todonotes}

\usepackage{cite}
\usepackage{graphicx}
\usepackage{subcaption}

\usepackage{xcolor}
\definecolor{light-blue}{rgb}{0.3,0.5,0.8}
\usepackage[colorlinks=true, linkcolor=cyan, citecolor=cyan, filecolor=magenta, urlcolor=cyan]{hyperref}
\usepackage{cleveref}

\title{\bf Integrated Routing and Intersection Control for Mixed Traffic}

\author{Filippos N. Tzortzoglou$^\dagger$,~\IEEEmembership{Student Member,~IEEE,}
Pengbo Zhu$^\dagger$,~\IEEEmembership{Member,~IEEE,}\\
and Andreas A. Malikopoulos,~\IEEEmembership{Senior Member,~IEEE}
\thanks{This research was supported in part by NSF under Grants CNS-2401007, CMMI-2348381, IIS-2415478, in part by MathWorks, and in part by the Swiss National Science Foundation (SNSF) through project P500-2\_235379.}
\thanks{Filippos N. Tzortzoglou and Pengbo Zhu are with the Civil and Environmental Engineering Department, Cornell University, Ithaca, NY, USA.}
\thanks{Andreas A. Malikopoulos is with the Applied Mathematics, Systems Engineering, Mechanical Engineering, Electrical \& Computer Engineering, and School of Civil \& Environmental Engineering, Cornell University, Ithaca, NY, USA (emails: \texttt{ft253,pz283,amaliko@cornell.edu}).}
}

\usepackage{amsthm}

\newtheorem{theorem}{Theorem}

\newtheorem{remark}{Remark}

\begin{document}

  \maketitle

 \thispagestyle{empty}
 \pagestyle{empty}
 
\begin{abstract}
The rapid development of cyber-physical systems is driving a transition toward mixed traffic environments comprising both human-driven and connected and automated vehicles (CAVs). This shift presents a unique opportunity to leverage the efficient operation of CAVs to improve overall network throughput. This paper introduces a hierarchical framework designed to bridge macroscopic routing optimization at the network level with microscopic vicinity control at signalized intersections. The upper layer utilizes aggregated traffic information to provide proactive routing guidance for CAVs, aiming to minimize total travel time. The lower layer leverages local vehicle states to jointly optimize traffic light phases and individual CAV trajectories, aiming to reduce intersection crossing delays and optimize energy consumption, respectively. The effectiveness of the proposed framework is validated through SUMO on the Sioux Falls benchmark network. Results demonstrate that the integration of these macroscopic and microscopic layers yields significantly better performance compared to applying either layer in isolation, significantly improving network throughput and reducing congestion.

\end{abstract}
\begingroup
\renewcommand{\thefootnote}{$\dagger$}
\footnotetext{\textbf{The authors contributed equally to this work.}}
\endgroup
\section{Introduction}
Recently, society has witnessed an increasing adoption of automation and connectivity in mobility systems. Connected and automated vehicles (CAVs) have emerged as a promising solution for improving safety, efficiency, and comfort \cite{guanetti2018control,tzortzoglou2026toward}. A growing body of literature highlights the potential benefits of CAVs from different perspectives, including vehicle routing \cite{bang2023optimal,wollenstein2021routing, Salazar2019congestion}, coordination at traffic junctions \cite{Malikopoulos2020,tzortzoglou2024feasibility,tan2026real,xu2022general,naderi2025lane}, and mobility equity \cite{bang2024emergingequity, salazar2024accessibility}. However, the transition toward full CAV adoption is not straightforward. Notably, studies suggest that a 100\% CAV penetration rate is unlikely to be achieved before 2060 \cite{alessandrini2015automated}. Consequently, there will be a transitional phase in which CAVs operate alongside human-driven vehicles (HDVs) on public roads. 

Vehicle routing in transportation networks constitutes a fundamental problem in CAV research \cite{rossi2018routing}. The general goal of routing problems is to assign efficient routes to vehicles based on their origins, destinations, and prevailing traffic conditions to reduce congestion. This problem has been extensively studied from multiple perspectives, including scenarios with full CAV penetration \cite{chu2017dynamic}, joint routing and trajectory planning \cite{Bang2022combined, bang2023optimal}, mixed-traffic environments \cite{Salazar2019congestion}, vehicle rebalancing strategies \cite{wollenstein2021routing, Zhu2024Coverage}, charging scheduling for electric vehicles\cite{liang2020mobility}, and mobility equity ~\cite{bang2024emergingequity}. 

Beyond routing, there has been growing interest over the past decade in the joint optimization of CAV trajectories and traffic signals in mixed-traffic environments at signalized intersections. The objective in such problems is to determine signal timings and CAV trajectories in the vicinity of intersections, based on the state of approaching vehicles while accounting for the unpredictable behavior of HDVs, to improve throughput and/or reduce energy consumption. Early studies on autonomous intersection management have demonstrated that such frameworks can significantly enhance intersection capacity \cite{dresner2007sharing}. Recently, several approaches have been proposed to address this problem, including bi-level optimization schemes in which CAV trajectories are optimized in response to traffic signals \cite{kamal2019development,tzortzoglou2025safe}, reinforcement learning methods \cite{guo2023cotv}, and fully integrated formulations where both signal phases and CAV trajectories are optimized simultaneously~\cite{tajalli2021traffic,suriyarachchi2023optimization,le2024distributed}. In our previous work \cite{tzortzoglou2025safe}, we proposed a bi-level framework that optimizes CAV trajectories and signal timings under mixed traffic conditions. A pressure-based signal control policy dynamically adjusts phase durations based on lane-level vehicle densities. In turn, CAVs plan their trajectories via optimal control to either cross the intersection or execute an energy-optimal stop. 

Although significant research exists on vehicle routing and the joint optimization of traffic signals and CAV trajectories, to the best of our knowledge, no prior work has merged these problems to explore the benefits of signal–CAV coordination at the network scale. The studies closest to this work are \cite{suriyarachchi2023optimization,liu2025integrated,niroumand2025real}. Specifically, in \cite{suriyarachchi2023optimization}, the authors consider the joint control of traffic signals and CAV trajectories in a multi-intersection setting; however, their evaluated network consists of only five intersections, and routing decisions are assumed to be exogenous. In \cite{liu2025integrated}, the authors propose a joint optimization framework for traffic signals, CAVs, and routing, but they assume a 100\% CAV penetration rate. Finally, in \cite{niroumand2025real}, the authors address the joint optimization problem for traffic signals and CAVs under mixed traffic conditions, but they restrict their focus to adjacent intersections without network-level routing.

This paper addresses the aforementioned gap by proposing a hierarchical framework that integrates system-level routing with the joint optimization of CAV trajectories and traffic signals at the intersection level. The upper layer determines real-time route choices for CAVs to minimize the total travel time, while the lower layer solves a joint optimization problem involving CAV trajectories and traffic signal timings in the vicinity of signalized intersections. The joint control of traffic signals and CAV trajectories is performed in a decentralized manner for each intersection, without assuming a global centralized planner that governs microscopic actions across the entire network. Furthermore, the two layers are coupled such that macroscopic routing decisions account for the estimated link travel times produced by the lower layer.

The remainder of this paper is organized as follows. In Section II, we propose the hierarchical framework. In Section III, we describe the lower-layer joint optimization framework that connects CAV trajectory control with traffic signal timing. In Section IV, we present the upper-layer routing problem and discuss the communication between the two layers. In Section V, we present the simulation results using SUMO on the Sioux Falls network. Finally, in Section VI, we draw concluding remarks and discuss future work.

\section{Hierarchical Framework}

In this section, we present the proposed hierarchical framework, detailing how we decompose the complex problem of jointly optimizing CAV routing, traffic signal timings, and microscopic CAV trajectories across a transportation network under mixed traffic conditions.

For multi-agent systems, a centralized controller may suffer from scalability issues and heavy computational burdens that make it impossible to implement unified optimization algorithms in real-time \cite{Lakshmikantham1981LargescaleDS}. To overcome this limitation, a well-received solution is to use a hierarchical framework and design a multilayer control scheme  \cite{YILDIRIMOGLU2018hierarchical}. In this setup, the upper-layer controller operates at a larger time step and interacts with a lower-layer controller which acts on local information to optimize local dynamics at a faster frequency. 

\begin{figure}[htpb]
    \centering
    \includegraphics[width=1\linewidth]{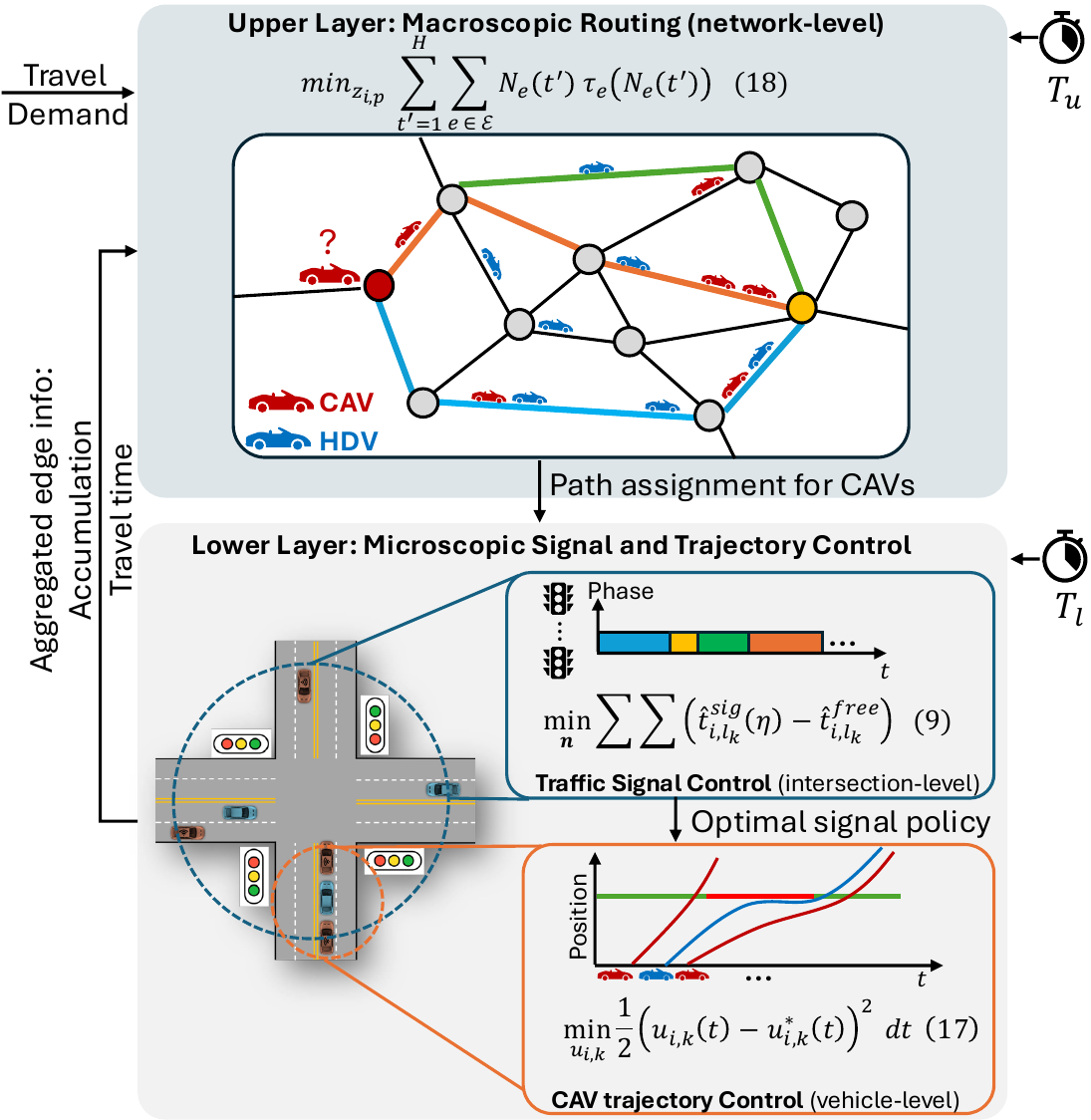}
    \caption{
Hierarchical control framework for CAVs in mixed traffic. For each CAV with a given origin and destination, the upper layer determines the route, with candidate routes shown in different colors. At each signalized intersection, the phase durations are displayed. The lower layer then designs the CAV trajectories (red curves).}
    \label{fig:hierframework}
    \vspace{-12pt}
\end{figure}

In this paper, we propose a hierarchical structure to alleviate congestion in the network under given demand patterns, as illustrated in Fig. \ref{fig:hierframework}. At the upper layer, the controller leverages aggregated information from local measurements, e.g., the number of vehicles and estimated travel times on each road segment. Its objective is to minimize the total travel time on each link of the network by providing route guidance for CAVs. On the other hand, the lower level focuses on the control of CAV trajectories at the vicinity of each signalized intersection, while optimizing the traffic signals durations.
\begin{figure*}
    \centering
    \includegraphics[width=0.68\linewidth]{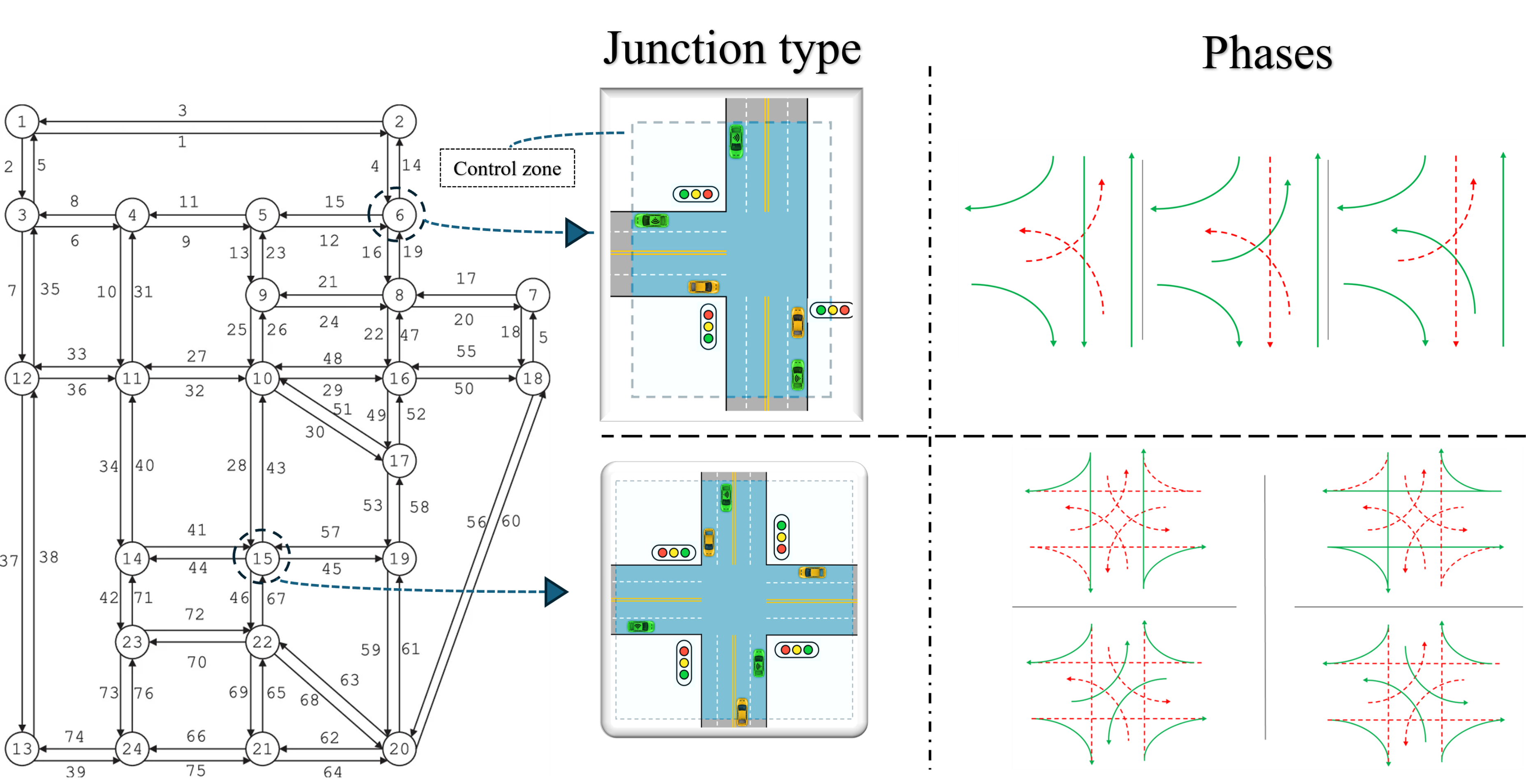}
    \caption{Sioux Falls Network: Junction Types and Signal Phases}
    \label{fig:network}
    \vspace{-13pt}
\end{figure*}

\section{Lower-level problem: Joint Optimization for CAVs and Traffic Signals}
\label{section_joint_cav_signals}
\subsection{Intersection Modeling}
Although our proposed framework can be easily generalized, we evaluate its performance using the benchmark Sioux Falls network \cite{ TransportationNetworksRepo}, as illustrated in Fig.~\ref{fig:network}. The original network, which consists of 24 nodes and 76 directed edges, has been explicitly adapted in this work to include microscopic signalized intersections.  We denote the network as
$\mathcal{G}=(\mathcal{K},\mathcal{E})$, where $\mathcal{K}=\{1,2,\dots, K\}$
is the set of nodes and $\mathcal{E}$ is the set of directed edges connecting
them. Each intersection $k$ has a control zone of range $R_k$, within which a coordinator (in our case, the smart traffic light infrastructure) can exchange information with CAVs and obtain the state of all vehicles. The state of HDVs can be obtained via roadside units. We define the set of incoming links of intersection $k$ as $\mathcal{E}_k \subseteq \mathcal{E}$, and denote by $\mathcal{C}_k$ the set of CAVs operating on these links within the control zone range $R_k$. We assume each link has two lanes and is bidirectional; for every edge $(o,d) \in \mathcal{E}$, the reverse edge $(d,o) \in \mathcal{E}$ also exists. Each node’s geometry is defined by its incoming and outgoing links. This network has 4 types of
nodes: 2-leg, 3-leg, 4-leg, and a 5-leg junction (see Fig.~\ref{fig:network} for illustration).

\subsection{Traffic Signal Control}
\label{sec:signal_control}

\subsubsection{Signal Structure and Policy}
\label{sec:signal_setup}

Each intersection $k\in\mathcal{K}$ with more than two
incoming--outgoing links is controlled by a traffic signal with
$M$~phases. The set of phases of intersection $k$ is denoted as $\mathcal{P}_k=\{1,\dots,M\}$.  During phase $m\in\mathcal{P}_k$, the links with right-of-way are denoted as $\mathcal{E}_{k,m}\subseteq\mathcal{E}_k$ (see Fig.~\ref{fig:network} for all phase configurations).

The phases repeat in a fixed order $\{1,2,\dots,M,1,2,\dots\}$. If no vehicle is detected on the links associated with a phase, that phase is allowed to have zero duration.
The optimization horizon covers two cycles,
$\mathcal{H}=\{1,\dots,2M\}$ and the phase active at step~$h$
is denoted as $\sigma(h)=((h{-}1)\!\mod M)+1$. For each step
$h\in\mathcal{H}$, let $\eta_h>0$ be the green duration of the associated phase $\sigma(h)$.  Every green period is followed by a clearance interval of fixed
duration~$h_\mathrm{c}$.  The green start
and end times are computed recursively:
\begin{align}
  \tau_1 &= 0,
  \qquad
  \theta_1 = \eta_1,
  \label{eq:timing_init}  \\
  \tau_h &= \theta_{h-1}+h_\mathrm{c},
  \qquad
  \theta_h = \tau_h + \eta_h,
  \quad h=2,\dots,2M.
  \label{eq:timing_recursion}
\end{align}
The signal policy of intersection~$k$ at time~$t$ is
\begin{equation}
  \bm{\phi}_k(t)
  =
  \bigl[\,
    \sigma_0(t),\;\eta_0(t);\;
    \bm{\sigma}^\top,\;\bm{\eta}^\top
  \,\bigr],
  \label{eq:signal_policy}
\end{equation}
where $\sigma_0(t)$ is the currently active phase, $\eta_0(t)$ its
remaining duration,
$\bm{\sigma}=[\sigma(1),\dots,\sigma(2M)]^\top$, and
$\bm{\eta}=[\eta_1,\dots,\eta_{2M}]^\top$.  The network-wide
policy is denoted as
$\bm{\Phi}(t)=\{\bm{\phi}_k(t)\}_{k\in\mathcal{K}}$.

\subsubsection{Estimation of Vehicle Crossing Times}
\label{sec:crossing_time_estimation}

Let $\mathcal{C}_{l_k}\subseteq\mathcal{C}_k$ denote the set of vehicles on link $l_k\in\mathcal{E}_k$ within the control zone of intersection $k$, indexed by $i=\{1,\dots,|\mathcal{C}_{l_k}|\}$, such that $i=1$ represents the vehicle closest to the intersection stop line. We denote by $V_c$ the reference crossing speed (i.e., the free-flow speed), by $h_f$ the saturation headway, and by $\pi(l_k)\in\mathcal{P}_k$ the signal phase associated with link $l_k$.

We next develop a metric to quantify signal-induced delay, defined as the difference between vehicle crossing times under a given signal control policy and the corresponding crossing times under uninterrupted (all-green) conditions. We begin by defining the free-flow crossing time.

\paragraph{Free-flow crossing time}
Considering all signals are green, the earliest crossing time of
vehicle~$i$ on link~$l_k$ is bounded by its free flow travel time and the
discharge headway of the preceding vehicle, which is estimated as:
\begin{equation}
  \hat{t}_{i,l_k}^{\,\mathrm{free}}
  =
  \max\!\Bigl(
    \frac{p_{i,l_k}}{V_\mathrm{c}}
      + \varepsilon(v_{i,l_k}),\;
    \hat{t}_{i-1,l_k}^{\,\mathrm{free}}
      + h_\mathrm{f}
  \Bigr),
  \label{eq:free_flow_crossing}
\end{equation}
where $\varepsilon(v_{i,l_k})\geq 0$ is a start-up delay that is
positive when the vehicle is near standstill and zero otherwise. The first term in \eqref{eq:free_flow_crossing} represents the travel-time-based
arrival estimate, determined by the vehicle’s current position $p_{i,l_k}$ and the free-flow speed $V_c$.
The second term enforces the headway constraint, ensuring that the
vehicle cannot cross earlier than its predecessor plus the saturation headway $h_\mathrm{f}$. For the lead vehicle ($i=1$), the headway is omitted. Formula \eqref{eq:free_flow_crossing} is applied to all vehicles in the signal optimization. Although actual trajectories may deviate under free-flow conditions, we adopt \eqref{eq:free_flow_crossing} for its closed-form structure, ability to capture typical traffic behavior, and computational efficiency in real time. We next focus on signal-induced delay.

\paragraph{Signalized crossing time}
Given a candidate signal schedule
$(\sigma(h),\tau_h,\theta_h)_{h\in\mathcal{H}}$, the signalized
crossing times are computed sequentially for all vehicles
$i=1,\dots,|\mathcal{C}_{l_k}|$, $\forall l_k \in \mathcal{E}_k$.

For each link $l_k$, we define the set of horizon steps during which its associated phase $\pi(l_k)$ is active as: $$\mathcal{H}_{l_k}=\{h\in\mathcal{H}:\sigma(h)=\pi(l_k)\}.
$$

To account for both signal timing and queue discharge, we define the earliest feasible crossing time of vehicle $i$ on link $l_k$, denoted by $r_{i,l_k}(\bm{\eta})$, as a function of $\bm{\eta}$:
\begin{equation}
  r_{i,l_k}(\bm{\eta})
  =
  \begin{cases}
    \hat{t}_{i,l_k}^{\,\mathrm{free}},
      & i=1,
    \\[4pt]
    \max\!\Bigl(
      \hat{t}_{i,l_k}^{\,\mathrm{free}},\;
      \hat{t}_{i-1,l_k}^{\,\mathrm{sig}}(\bm{\eta}) + h_\mathrm{f}
    \Bigr),
      & i\geq 2.
  \end{cases}
  \label{eq:reference_time_signalized}
\end{equation}

The estimated crossing time is then
\begin{equation}
  \hat{t}_{i,l_k}^{\,\mathrm{sig}}(\bm{\eta})
  =
  \begin{cases}
    r_{i,l_k}(\bm{\eta}),
      \quad \text{if }\exists\, h\in\mathcal{H}_{l_k}:\;
        \tau_h \le r_{i,l_k}(\bm{\eta}) \le \theta_h,
    \\[5pt]
    \displaystyle \tau^* ,
      \quad \quad \quad \text{if }\exists\, h\in\mathcal{H}_{l_k}:\;
        \tau_h > r_{i,l_k}(\bm{\eta}),
    \\[8pt]
    \theta_{2M}+\Delta(\pi(l_k)),
      \quad \text{otherwise},
  \end{cases}
  \label{eq:signalized_crossing_time}
\end{equation}
where
$
    \tau^*= \min\bigl\{\tau_h \,:\, h\in\mathcal{H}_{l_k},\;
    \tau_h > r_{i,l_k}(\bm{\eta})\bigr\},
$
and
$
  \Delta(m)
  =
  \bar{w}_{l_k}\,\bigl[\bigl(m-\sigma(2M)\bigr)\!\!\mod M\bigr],
  \label{eq:beyond_horizon_offset}
$
with $\bar{w}_{l_k}$ denoting the average green duration associated with link $l_k$, used to approximate the signal evolution beyond the horizon. Note that \eqref{eq:reference_time_signalized} differs from \eqref{eq:free_flow_crossing}, since the predecessor’s crossing time is determined by signalized operation rather than by its free-flow estimate.

Thus, in \eqref{eq:signalized_crossing_time}, the reference time $r_{i,l_k}(\bm{\eta})$ first accounts for both the vehicle’s free-flow arrival and the minimum headway behind its predecessor. If this reference time falls within a green interval of phase $\pi(l_k)$, the vehicle crosses immediately (branch 1). Otherwise, the vehicle is assigned to the start of the next available green interval of phase $\pi(l_k)$ within the optimization horizon (branch 2). If no such green interval exists within the horizon, the crossing time is approximated by projecting the phase sequence beyond the horizon using the average green duration $\bar{w}_{l_k}$ (branch 3).

Finally, the estimated delay for vehicle $i$ is defined as the difference between its actual signalized crossing time and its free-flow crossing time:
\begin{equation}
  \hat{d}_{i,l_k}(\bm{\eta})
  =
  \hat{t}_{i,l_k}^{\,\mathrm{sig}}(\bm{\eta})
  -
  \hat{t}_{i,l_k}^{\,\mathrm{free}}.
  \label{eq:estimated_delay}
\end{equation}

\subsection{Signal Optimization Problem}
\label{sec:signal_opt_problem}

Having established the delay estimation framework, we now formulate the signal optimization problem. Rather than optimizing the green durations directly, we parameterize the signal phase control policy through the discrete number of vehicles expected to discharge during each phase.

Specifically, for each link $l_k \in \mathcal{E}_k$ and horizon step $h \in \mathcal{H}$, let $n_{l_k,h} \in \mathbb{Z}_{\geq 0}$ denote the number of vehicles from link $l_k$ scheduled to discharge during step $h$. Since each step $h$ corresponds to a signal phase, the green duration is determined by the most heavily loaded served link. Hence, the green duration $\eta_h$ at step $h$ is given by
\begin{equation}
  \eta_h
  =
  \Bigl(\max_{l_k\in\mathcal{E}_{k,\sigma(h)}} n_{l_k,h}\Bigr)
  h_\mathrm{f}
  + \varepsilon_h,
  \label{eq:green_from_n}
\end{equation}
where $h_\mathrm{f}$ is the saturation headway, $\varepsilon_h$ denotes the corresponding lost time, and $\mathcal{E}_{k,\sigma(h)} \subseteq \mathcal{E}_k$ is the subset of links served by the active phase $\sigma(h)$.

Let $\bm{n}=(n_{l_k,h})_{l_k\in\mathcal{E}_k,h\in\mathcal{H}}$ collect all decision variables, and let $\bm{\eta}(\bm{n})$ denote the vector of green durations induced by \eqref{eq:green_from_n}. The objective is to minimize the total estimated signal-induced delay of all vehicles approaching the intersection over the optimization horizon. To avoid interruption with CAV trajectories, the first cycle, i.e., the first $M$ steps of the horizon, is kept fixed at previously committed values, and only the second cycle $h \in \{M+1,\dots,2M\}$ is optimized. The resulting optimization problem is formulated as
\begin{subequations}\label{eq:signal_opt}
\begin{align}
  \min_{\bm{n}} \;\;
  & \sum_{l_k\in\mathcal{E}_k}
    \sum_{i\in\mathcal{C}_{l_k}}
    \hat{d}_{i,l_k}\!\bigl(\bm{\eta}(\bm{n})\bigr)
  \label{eq:signal_obj}\\[3pt]
  \textrm{s.t.}\;\;
  & n_{l_k,h} = 0,
    \qquad
    \forall\, l_k \notin \mathcal{E}_{k,\sigma(h)},\;
    \forall\, h \in \mathcal{H},
  \label{eq:phase_compatibility}\\
  & 0 \leq n_{l_k,h} \leq n_{\max},
    \qquad
    \forall\, l_k \in \mathcal{E}_{k,\sigma(h)},\;
    \forall\, h \in \mathcal{H},
  \label{eq:capacity_bound}\\
  & n_{l_k,h} = \bar{n}_{l_k,h}^{\,\mathrm{fix}},
    \qquad
    \forall\, l_k \in \mathcal{E}_k,\;
    h=1,\dots,M,
  \label{eq:fixed_first_cycle}\\
  & n_{l_k,h} \in \mathbb{Z}_{\geq 0},
    \qquad
    \forall\, l_k \in \mathcal{E}_k,\;
    \forall\, h \in \mathcal{H}.
  \label{eq:integer_constraint}
\end{align}
\end{subequations}

Since phase durations are induced by $\bm{n}$ through \eqref{eq:green_from_n}, problem \eqref{eq:signal_opt} determines the green allocation by selecting how many vehicles are served in each phase. Constraint \eqref{eq:phase_compatibility} ensures vehicles are assigned only to links served by the active phase. Constraint \eqref{eq:capacity_bound} imposes a per-step upper bound on discharged vehicles, while \eqref{eq:fixed_first_cycle} preserves the committed plan over the first cycle. Finally, \eqref{eq:integer_constraint} enforces integer discharge decisions.
%

\begin{algorithm}[H]
\caption{Receding hor. signal control at phase change}
\label{alg:receding_horizon}
\begin{algorithmic}[1]
\State Measure $(p_{i,l_k}, v_{i,l_k})$ for all vehicles in $\mathcal{C}_k$
\State Fix $n_{l,h}=\bar{n}_{l,h}^{\mathrm{fix}}$, \ $\forall h=1,\dots,M$
\State Solve \eqref{eq:signal_opt} for $n_{l,h}$, \ $\forall h=M+1,\dots,2M$
\State Broadcast $(\sigma(h),\tau_h,\theta_h)_{h\in\mathcal{H}}$ to all vehicles
\State CAVs plan trajectories via low-level control
\State Set $\bar{n}_{l,h}^{\mathrm{fix}} \gets n_{l,h}$, $\forall h=1,\dots,M$
\end{algorithmic}
\end{algorithm}
\vspace{-0.5em}

This receding-horizon problem is solved at each step $h$ (that is at the beginning of each phase), and adapts to evolving traffic conditions while providing CAVs with a stable schedule over the next cycle, as shown in Algorithm~1.

\subsection{Control of CAVs}

Subsequently, given the resulting signal plan, we compute the control trajectories for CAVs. We model the motion of each vehicle as a double integrator. In particular, consider a vehicle $i$ approaching intersection $k \in \mathcal{K}$ traveling on link $l_k\in\mathcal{E}_k$. Then, its state is described by the dynamics
\begin{align}
    \dot{p}_{i,l_k}(t) &= v_{i,l_k}(t), \nonumber \\
    \dot{v}_{i,l_k}(t) &= u_{i,l_k}(t),
    \label{eq:vehicle_dynamics}
\end{align}
where $p_{i,l_k}(t)$ denotes the position of vehicle $i$ on link $l_k\in \mathcal{E}_k$ (measured from the crossing point), $v_{i,l_k}(t)$ is its speed, and $u_{i,l_k}(t)$ is the acceleration.
The state and control variables are subject to the constraints
\begin{align}
0 \le v_{i,l_k}(t) &\le v_{\max}, 
\qquad i \in \mathcal{C}_k,\; l_k \in \mathcal{E}_k,
\label{eq:speed_constraints} \\
u_{i,\min} \le u_{i,l_k}(t) &\le u_{i,\max}, 
\qquad i \in \mathcal{C}_k,\; l_k \in \mathcal{E}_k.
\label{eq:control_constraints}
\end{align}
To ensure collision-free motion, rear-end safety constraints must be satisfied. For a vehicle $i$ and its preceding vehicle $j$ on the same link, the inter-vehicle distance must satisfy
\begin{equation}\label{rear_end_constraint}
p_{j,l_k}(t) - p_{i,l_k}(t) \geq s_0 + v_{i,l_k}(t)T,
\end{equation}
where $s_0$ denotes the standstill safety distance and $T$ represents the safe time headway. Also, each CAV must respect the traffic signal constraint at the intersection, formulated as:
\begin{align}
    t_{i,{l_k}}^{\text{cr}} \in \mathcal{T}_{i,l_k}(\phi_k(t)),
\end{align}
where $t_{i,l_k}^{\text{cr}}$ denotes the time when CAV $i$ on link $l_k$ crosses the stop line of intersection $k$. The set $\mathcal{T}_{i,l_k}(\phi_k(t))$ defines the available green intervals under the signal policy $\phi_k(t)$.

Consider now a CAV $i$ entering the control zone of intersection $k$ at time $t_{i,l_k}^0$. Through V2X communication, the vehicle receives the set of admissible crossing times $\mathcal{T}_{i,l_k}(\phi_k(t))$, as well as the planned crossing times of the rest of CAVs and the predictions of the crossing times of HDVs approaching the same intersection $k$. Then, the goal of each CAV is to compute a time-optimal reference trajectory to cross the intersection subject to the aforementioned constraints. 
\begin{remark}
    In networks with $100\%$ CAV penetration, the crossing times of other vehicles can be treated as deterministic. In mixed traffic, however, the crossing times of HDVs must be estimated, and the actual ones may deviate from their predicted values.\end{remark}
    \begin{figure}
        \centering
        \includegraphics[width=1.04\linewidth]{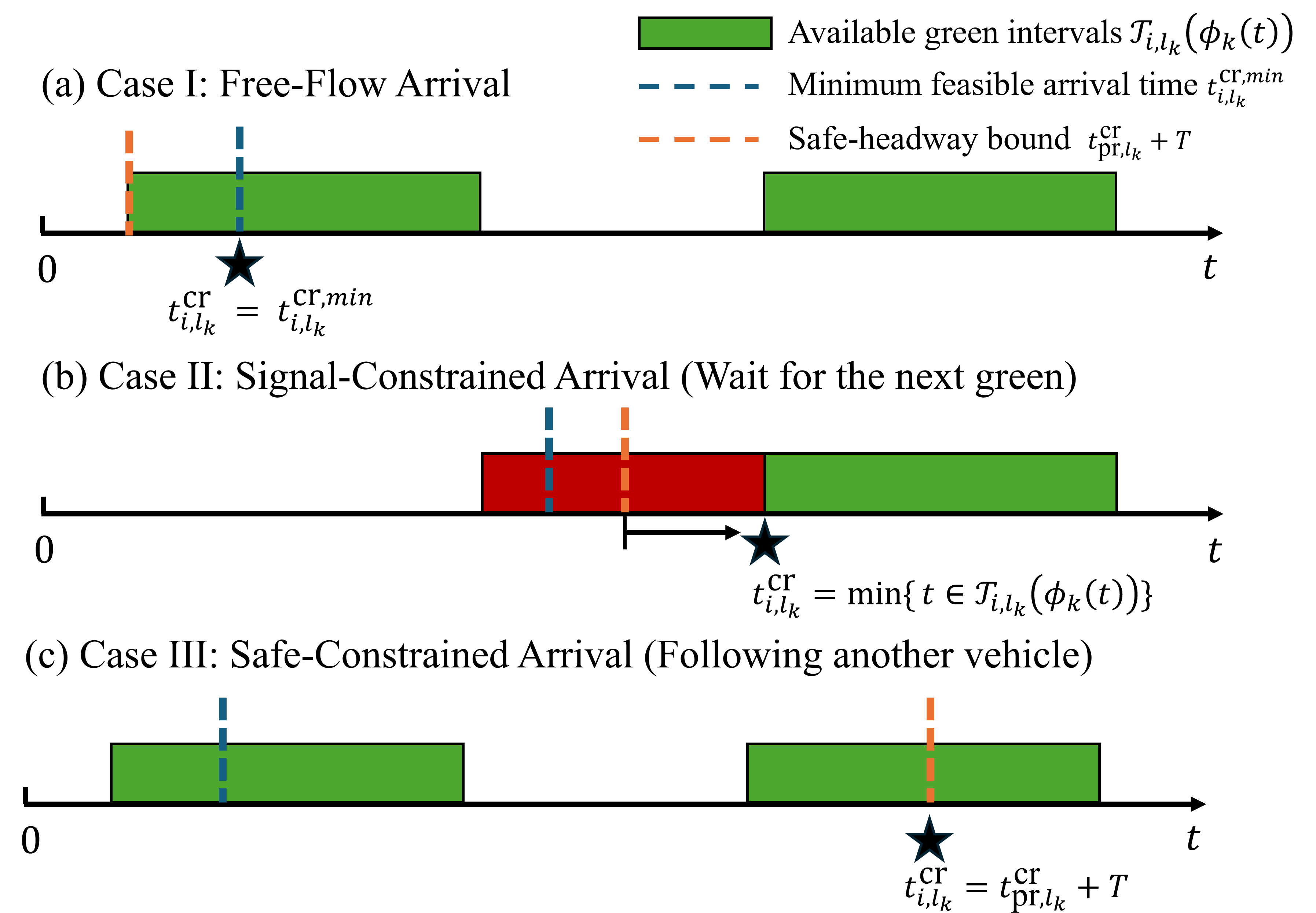}
        \caption{Illustration of the crossing time selection logic Eq.\ref{crossing_time_selection}. The selected crossing time $t_{i,l_k}^{\text{cr}}$($\bigstar$) is the earliest moment on or after the maximum of two bounds 1) the earliest arrival time $t_{i,l_k}^{{\text{cr}},\min}$ and 2) the safety-headway bound $t_{\text{pr},l_k}^{\text{cr}}+T$ that falls within an available green interval.}
        \label{fig:timeSelection}
        \vspace{-12pt}
    \end{figure}
Following \cite{tzortzoglou2025safe}, we employ the Intelligent Driver Model (IDM) to simulate the evolution of HDV trajectories forward in time to estimate their crossing times. The desired crossing time $t_{i,{l_k}}^{\text{cr}}$ of the CAV is determined by
\begin{align} \label{crossing_time_selection}
    t_{i,l_k}^{\text{cr}} = \min \Big\{ t \in \mathcal{T}_{i,l_k}(\phi_k(t)) \;\Big|\; t \geq \max\!\big(t_{i,l_k}^{{\text{cr}},\min},\, t_{\text{pr},l_k}^{\text{cr}}+T\big) \Big\},
\end{align}
where $t_{i,l_k}^{{\text{cr}},\min}$ denotes the minimum feasible arrival time of CAV $i$ at the stop line, subject to its speed and acceleration limits \eqref{eq:speed_constraints} and \eqref{eq:control_constraints}. The variable $t_{\text{pr},l_k}^{\text{cr}}$ denotes the crossing time of the preceding vehicle (which is estimated if the preceding vehicle is an HDV), and $T$ denotes the designated time headway. For reference, see Fig. \ref{fig:timeSelection}.

Given the desired crossing time $t_{i,l_k}^{\text{cr}}$, CAV $i$ computes its trajectory by solving the following optimal control problem:\\
\noindent \textbf{Energy-optimal control problem:} \label{prb:ocp-1}
\begin{equation}
\begin{aligned}
\label{eq:energy_cost}
&\underset{u_{i,l_k}}{\min} \quad \frac{1}{2} \int_{t^{0}_{i,l_k}}^{t_{i,l_k}^{\text{cr}}} u_{i,l_k}^2(t) \, \mathrm{d}t, \\
&\text{subject to:} \quad 
\eqref{eq:vehicle_dynamics}, \quad \eqref{eq:speed_constraints}, 
\quad \eqref{eq:control_constraints}, \\
&\text{given:} \quad p_{i,l_k} (t_{i,l_k}^0) = p_{i,l_k}^0, \quad v_{i,l_k} (t_{i,l_k}^0) = v_{i,l_k}^0, \\
&\quad \quad \quad  \; p_{i,l_k} (t_{i,l_k}^{\text{cr}}) = p_{i,l_k}^{\text{cr}}.
\end{aligned}
\end{equation}
Problem \eqref{eq:energy_cost} can be solved analytically using Pontryagin's Minimum Principle, yielding closed-form expressions for the optimal trajectories as shown in \cite{Malikopoulos2020}. The final trajectories take the following form: 
\begin{align} 
    &u_i(t) = 6 \phi_{i,3}t + 2\phi_{i,2}, \nonumber \\
    &v_i(t) = 3\phi_{i,3}t^2 + 2\phi_{i,2}t + \phi_{i,1},\label{polynomials} \\
    &p_i(t) = \phi_{i,3}t^3 + \phi_{i,2}t^2 + \phi_{i,1}t + \phi_{i,0}, \nonumber
\end{align}
where $\phi_{i,3},\phi_{i,2},\phi_{i,1},\phi_{i,0} \in \mathbb{R}$ are constants of integration, that can be found by the initial and terminal conditions and the optimal terminal acceleration $u(t_i^f)=0$ (see \cite{chalaki2021CSM}). 

Note that the rear-end safety constraint \eqref{rear_end_constraint} is not explicitly enforced in \eqref{eq:energy_cost}. In mixed traffic, where CAVs coexist with HDVs, enforcing this constraint at the planning stage does not ensure its satisfaction during execution due to HDV uncertainty. Thus, safety is enforced reactively by tracking the trajectory from \eqref{eq:energy_cost} under the rear-end constraint \eqref{rear_end_constraint}, yielding:

\noindent \textbf{Reactive safety tracking problem:}
\begin{equation}
\begin{aligned}
\label{eq:tracking}
&\underset{u_{i,l_k}}{\min} \quad \frac{1}{2} \big(u_{i,l_k}(t) - u_{i,l_k}^{*}(t)\big)^2, \\
&\text{subject to:} \quad \eqref{eq:speed_constraints}, \quad \eqref{eq:control_constraints}, \\
& \quad \quad \quad \quad \;\; \dot{h}_{i,l_k}(t) + \kappa \, h_{i,l_k}(t) \geq 0,
\end{aligned}
\end{equation}
where $u_{i,l_k}^{*}(t)$ is the reference control input obtained from \eqref{eq:energy_cost}, and $h_{i,l_k}(t) = p_{j,l_k}(t) - p_{i,l_k}(t) - s_0 - v_{i,l_k}(t)T$ is a safety function encoding the rear-end constraint \eqref{rear_end_constraint}, with $\kappa > 0$ a tunable parameter that governs the rate at which the safety margin is enforced. The constraint $\dot{h}_{i,l_k}(t) + \kappa \, h_{i,l_k}(t) \geq 0$ ensures that whenever the system approaches the safety boundary, the control input is adjusted to maintain a safe following distance. Since the objective is quadratic in $u_{i,l_k}$ and the safety constraint is linear with respect to $u_{i,l_k}$ after substituting the vehicle dynamics \eqref{eq:vehicle_dynamics}, this formulation reduces to a quadratic programming (see \cite{xu2022general}).

\subsection{Case of replanning}

When the deviation from the reference trajectory becomes significant, e.g., due to an unexpected deceleration of a preceding HDV, the CAV may no longer be able to cross within its assigned green interval. To address this, we employ an event-triggered replanning mechanism. Specifically, if the tracking error exceeds a threshold $\epsilon$, the CAV re-estimates the preceding vehicle’s crossing time and re-checks the feasibility of its own crossing within the current green interval. If the originally assigned crossing time is no longer feasible, the CAV updates it by re-solving \eqref{crossing_time_selection} using the updated estimate $t_{\text{pr},l_k}^{\text{cr}}$. It then recomputes the reference trajectory by solving \eqref{eq:energy_cost} and updates the reactive tracking problem \eqref{eq:tracking}.
\section{Upper level problem: Routing for CAVs}
\label{sec:high_level_controller}

In the previous section, we discussed how to design traffic signal control and trajectory optimization for CAVs, minimizing travel delays when vehicles approach intersections. In this section, we shift our focus to the macroscopic level and investigate how to optimize the route choices of CAVs to alleviate network congestion while accounting for the uncontrollable route choices of HDVs.

\subsection{Problem Formulation}


Let $N_e(t')$ denote the accumulation of vehicles on edge $e \in \mathcal{E}$ at macroscopic time step $t'$, and $\tau_e(N_e(t'))$ denote the corresponding travel time.

The high-level objective is to minimize the Total Travel Time (TTT) over the planning horizon $H$. The objective is 
\begin{equation}
\label{obj: TTT}
    J = \sum_{t'=1}^{H} \sum_{e \in \mathcal{E}} N_e(t') \tau_e(N_e(t')).
\end{equation}
For a given origin-destination pair, let $\mathcal{P}_{o,d}$ be the set of $P$ candidate paths. In practice, we can pre-compute a set of candidate routes using the $k$-shortest path algorithm. For each CAV $i$ departing at time $t'$, we define a binary decision variable as:
\begin{equation}
z_{i,p} = 
\begin{cases} 
1 & \text{if vehicle } i \text{ is assigned to path } p \in \mathcal{P}_{o,d}, \\
0 & \text{otherwise}.
\end{cases}
\end{equation}
The routing assignment is constrained such that each vehicle selects only one path, ensuring that $\sum_{p \in \mathcal{P}_{o,d}} z_{i,p} = 1$.

\subsection{Marginal Cost as Discrete Gradient Descent}
Solving \eqref{obj: TTT} in a centralized way is computationally intractable for real-time large-scale networks. To this end, we propose a decentralized marginal cost-based routing policy that acts as a gradient descent algorithm on the global network delay.  For a controllable CAV $i$, the marginal cost (MC), i.e., the perturbation brought by adding a single vehicle onto edge $e$ at the time of arrival, is defined as:
\begin{equation}
    MC_e =  \tau_e(N_e) + N_e \frac{\partial \tau_e(N_e)}{\partial N_e},
\end{equation}
where $N_e(t') = N_{e, HDV}(t') + N_{e, CAV}(t')$.

\begin{theorem}
Given an exogenous, HDV flow distribution, assigning a CAV sequence to paths that minimizes the $MC_e$ strictly minimizes the TTT objective.
\end{theorem}

\begin{proof}
The gradient of the objective $J$ with respect to the controllable flow on edge $e$ at time $t'$ is:
$$\frac{\partial J}{\partial \hat{N}_{e, CAV}(t')} = \frac{\partial}{\partial N_e(t')} \Big[ N_e(t') \tau_e(N_e(t')) \Big] \cdot \frac{\partial N_e(t')}{\partial \hat{N}_{e, CAV}(t')}.$$
The term $\frac{\partial N_e(t')}{\partial \hat{N}_{e, CAV}(t')} = 1$. Applying the product rule to the remaining term yields
$$\frac{\partial J}{\partial \hat{N}_{e, CAV}(t')} = \tau_e(N_e(t')) + N_e(t') \frac{\partial \tau_e(N_e(t'))}{\partial N_e(t')}.$$

Consequently, the perturbation to the global objective $\Delta J_p$ caused by assigning vehicle $i$ to path $p$ is approximated by the sum of the directional derivatives along that path:
$$\Delta J_p \approx \sum_{e \in p} \left( \tau_e(N_e) + N_e \frac{\partial \tau_e}{\partial N_e} \right) = \sum_{e \in p} MC_e(t'),$$
which represents the total marginal impact of path $p$ on the network. By constructing the decision vector $z_{i,p}$ such that $z_{i,p^*} = 1$ for $p^* = \arg\min_p \sum_{e \in p} MC_e(t')$, it ensures that the controller selects the routing assignment with the steepest descent direction with respect to the global objective $J$.
\end{proof}

\subsection{Measurement-Based Implementation}
Due to the fact that the travel time function $\tau_e(\cdot)$ is governed by stochastic microscopic car-following dynamics and signal control, an explicit analytical derivative of it is unavailable. Instead, to compute the $MC_e$ in practice without relying on physics models, the controller leverages empirical data from the low-level infrastructure. At each macroscopic time step $t'$, the controller receives the measured travel time $\tilde{\tau}_e(t')$ and accumulation $\tilde{N}_e(t')$ from the microscopic network. We estimate congestion sensitivity using a moving-window finite difference of simulator feedback given by
$$\eta_e(t') = \frac{\partial \tau_e}{\partial N_e} \approx \frac{\tilde{\tau}_e(t') - \tilde{\tau}_e(t' - \Delta t')}{\tilde{N}_e(t') - \tilde{N}_e(t' - \Delta t')}.$$
To ensure numerical stability, $\eta_e(t')$ is subjected to a low-pass filter to smooth high-frequency microscopic noise.

Future network states are estimated using spatiotemporal memory matrices $\mathcal{M}^{HDV} \in \mathbb{R}^{|\mathcal{E}| \times H}$ and $\mathcal{M}^{CAV} \in \mathbb{R}^{|\mathcal{E}| \times H}$. As vehicles are routed, their anticipated physical presence on future edges is reserved in these matrices based on accumulated travel times. Therefore, the predicted future volume on edge $e$ at future time step $m$ is $N_e(m) = \mathcal{M}^{HDV}_{e,m} + \mathcal{M}^{CAV}_{e,m}$.

When evaluating a path $p$ for a CAV, the controller traverses the path edge-by-edge, accumulating the projected time of arrival $m$ for each subsequent edge. Using the predicted future volume $N_e(m)$ and the empirically derived congestion sensitivity $\eta_e$, we project the future travel time via a first-order Taylor expansion
$$\hat{\tau}_e(m) = \max \Big( \tau^{free}_e, \, \tilde{\tau}_e(t') + \eta_e \big( N_e(m) - \tilde{N}_e(t') \big) \Big),$$
where $\tau^{free}_e$ is the free-flow travel time upper bound. 

Finally, the predicted MC for the path, $\hat{MC_p}$, is computed by substituting these empirical predictions into our gradient descent formulation
$$\hat{MC_p} = \sum_{e \in p} \Big( \hat{\tau}_e(m) + N_e(m) \eta_e \Big).$$
The CAV is then assigned to the path that minimizes $\hat{MC_p}$, and its future trajectory is registered into the memory matrix $\mathcal{M}^{CAV}$ to inform subsequent routing decisions.
\begin{figure}[t]
    \centering
    \includegraphics[width=0.9\linewidth]{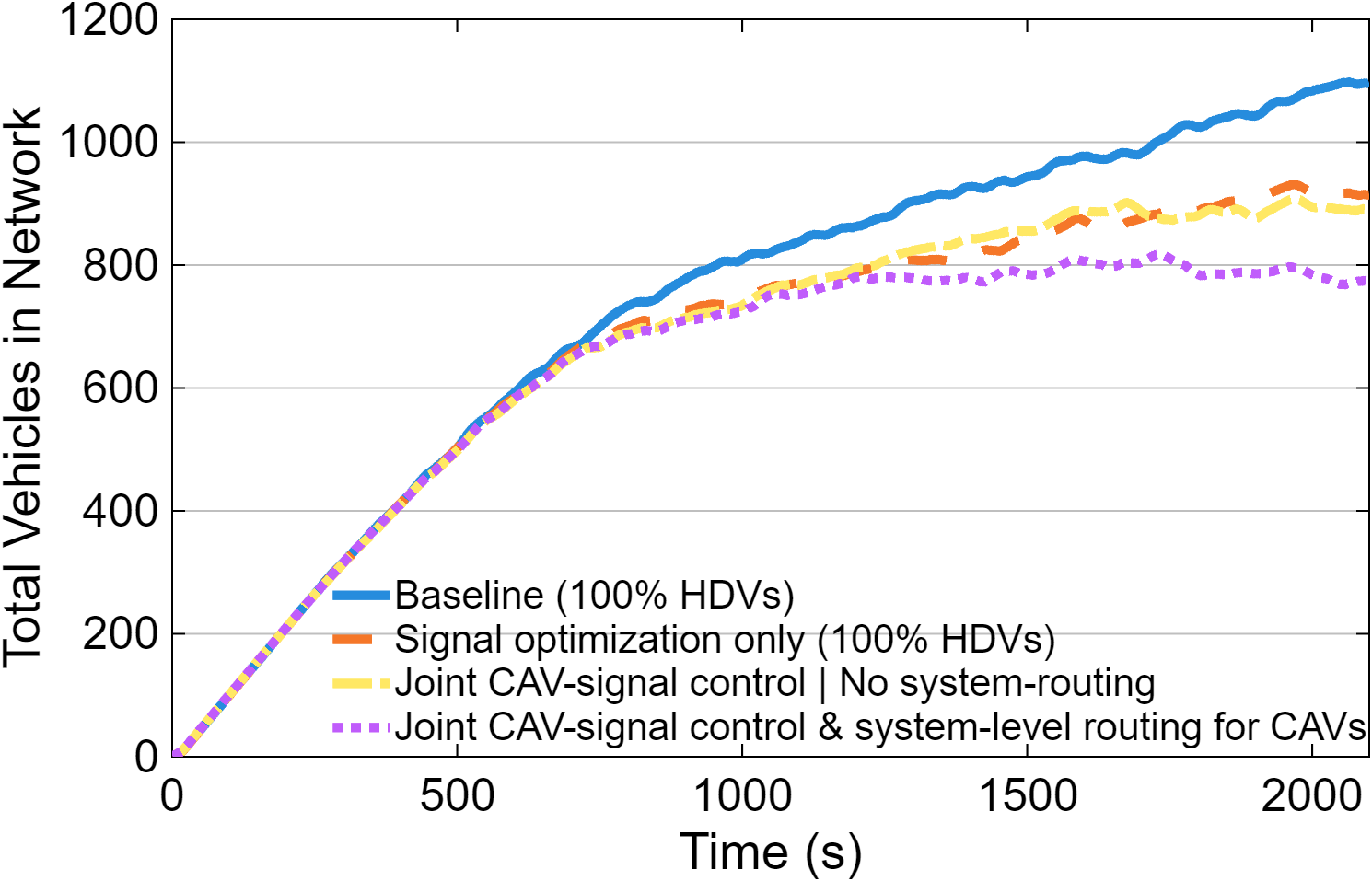}
    \caption{Evolution of total vehicles in the network}
    \label{fig:Total vehicles in the network as a function of time}
\end{figure}

\begin{figure*}[t]
    \centering
    \begin{subfigure}[t]{0.46\textwidth}
        \centering
        \includegraphics[width=\linewidth]{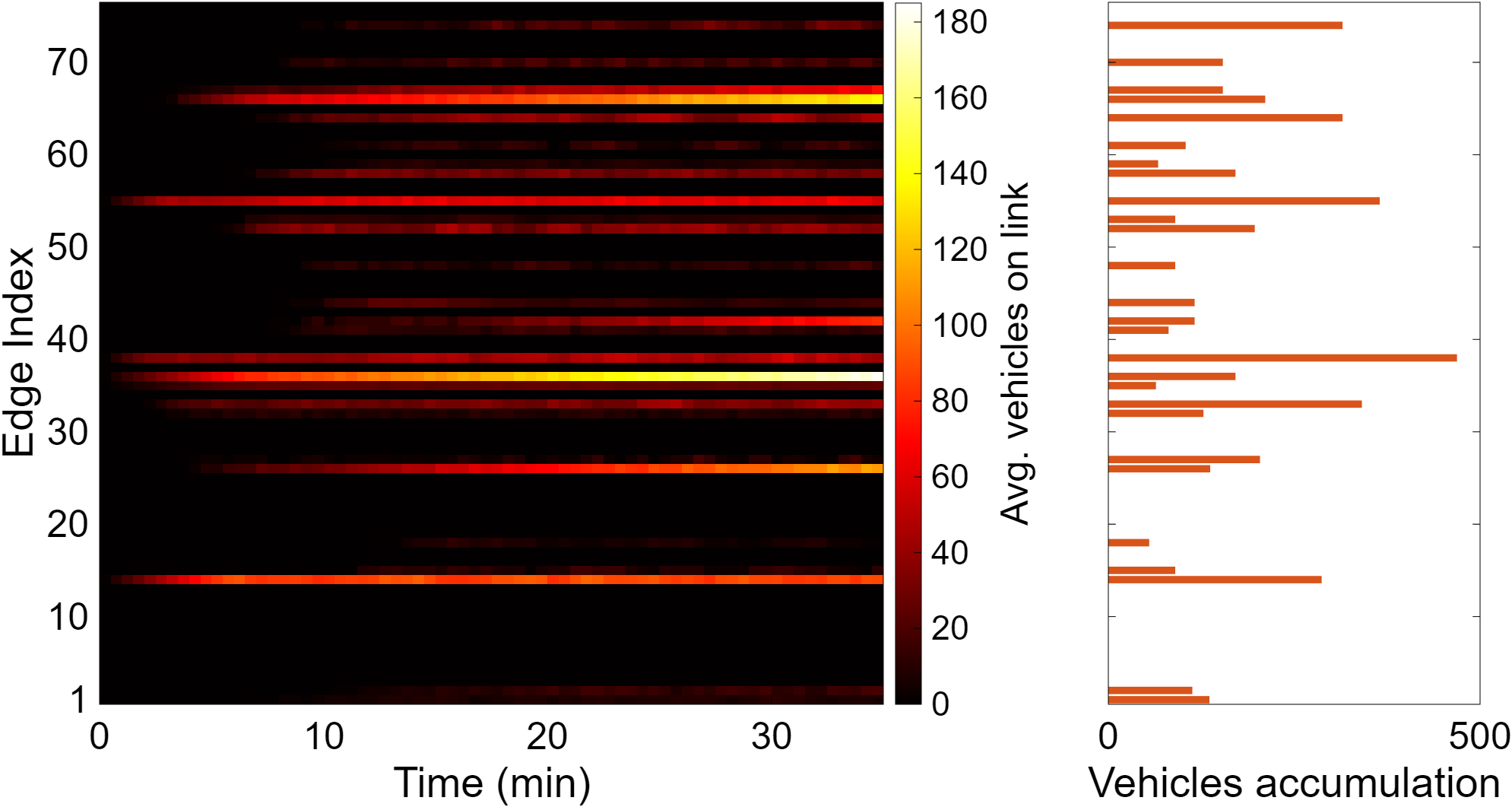}
        \caption{Baseline (100\% HDVs)}
        \label{flow_100HDVs_nothing_optimized}
    \end{subfigure}
    \hfill
    \begin{subfigure}[t]{0.46\textwidth}
        \centering
        \includegraphics[width=\linewidth]{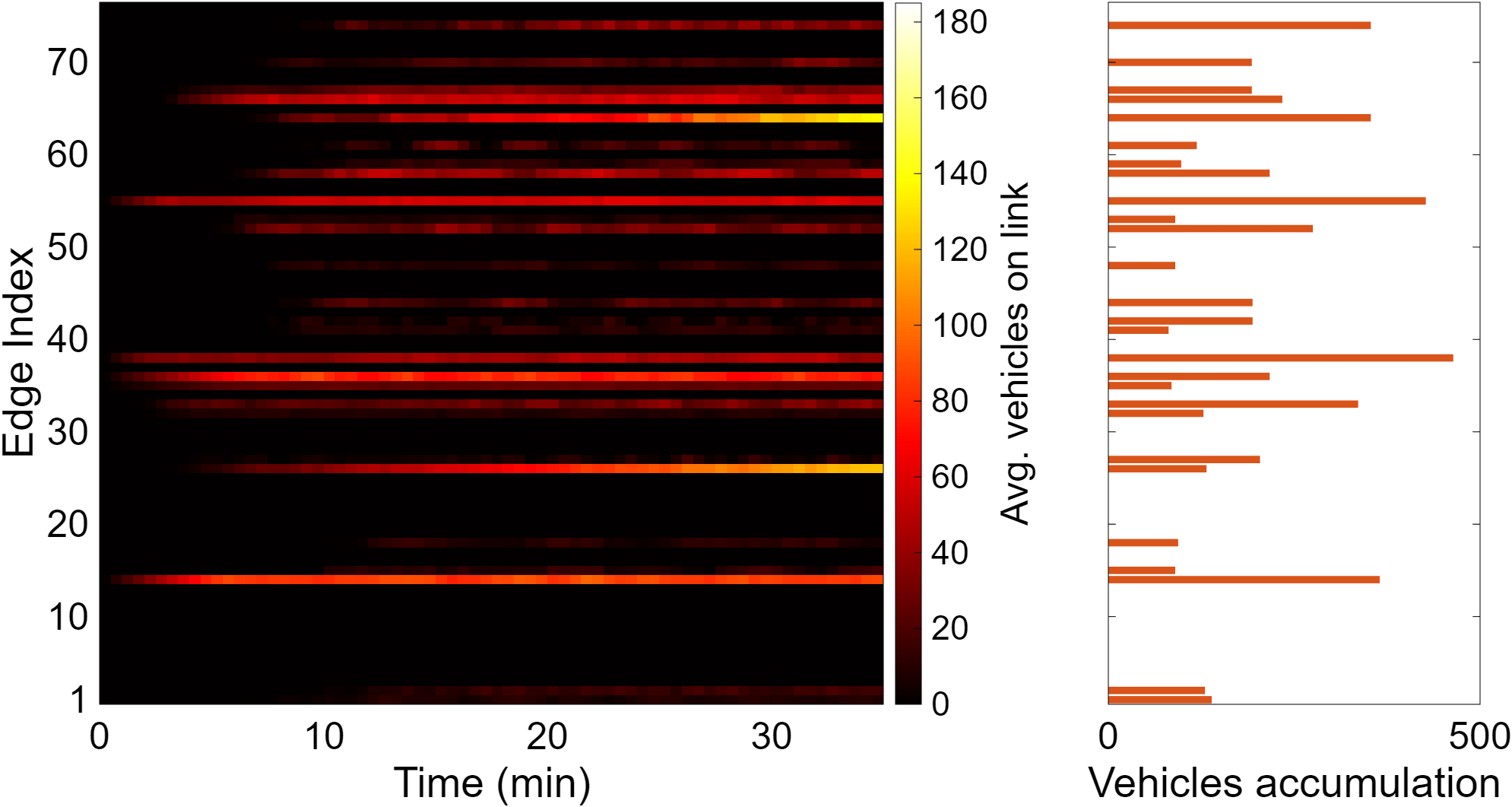}
        \caption{Signal optimization only (100\% HDVs)}
                \label{flow_100HDVs_signals_optimized}
    \end{subfigure}

    \vspace{0.5em} 

    \begin{subfigure}[t]{0.46\textwidth}
        \centering
        \includegraphics[width=\linewidth]{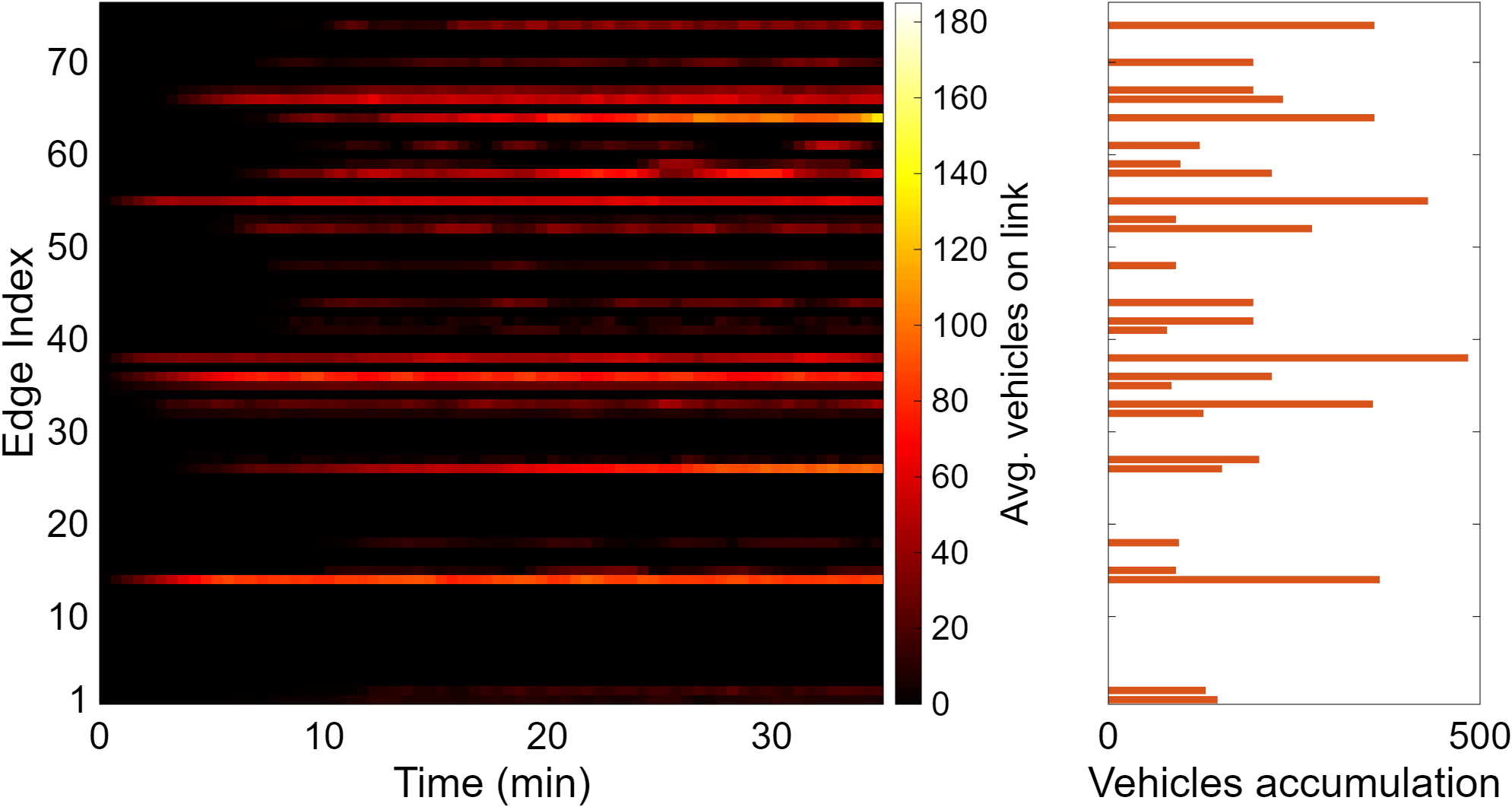}
        \caption{Joint CAV-signal control $|$ No system-level-routing}
                        \label{flow_50HDVs_signals_optimized_and_CAV_optimized_no_routing}
    \end{subfigure}
    \hfill
    \begin{subfigure}[t]{0.46\textwidth}
        \centering
        \includegraphics[width=\linewidth]{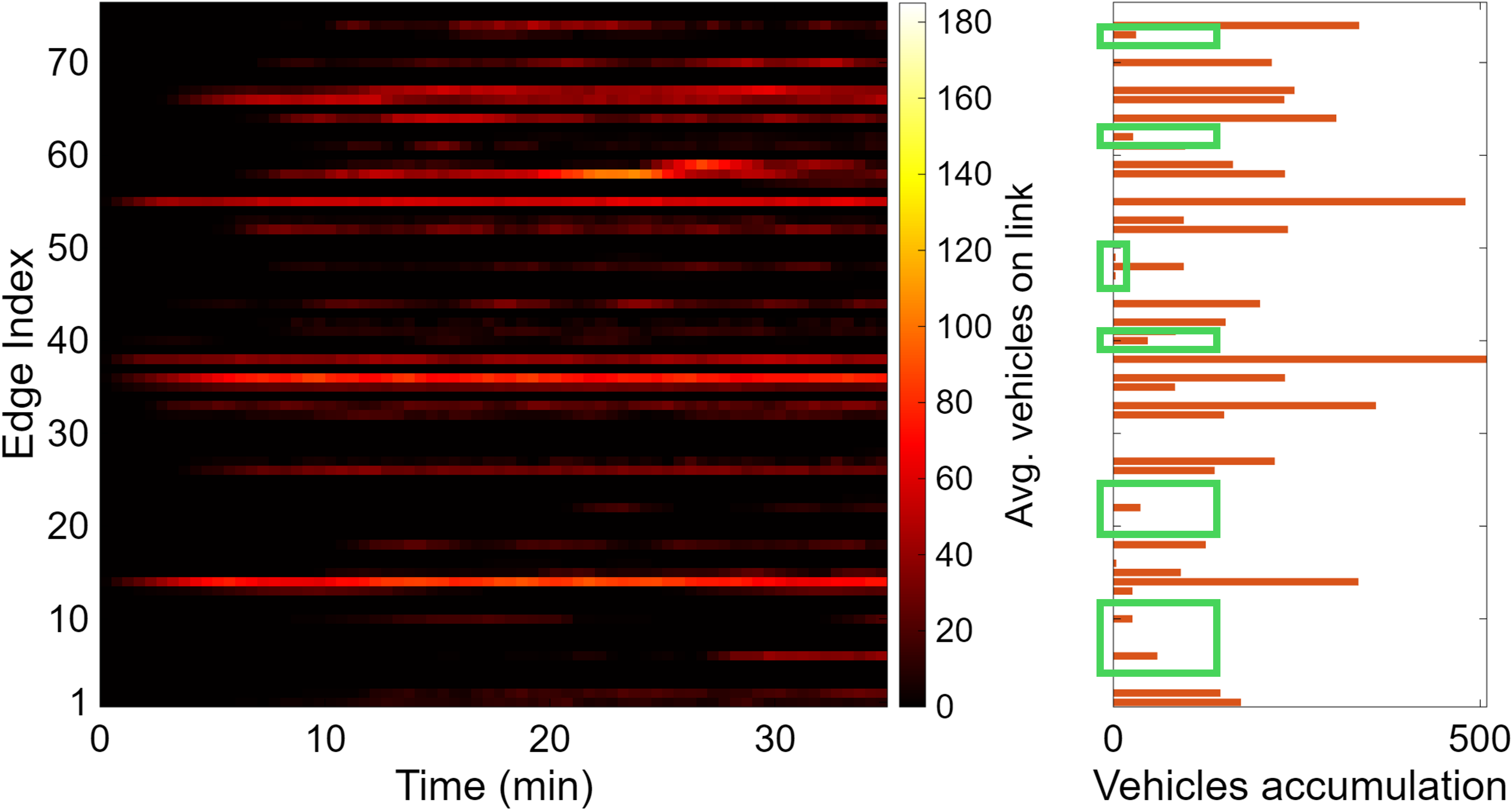}
\caption{Joint CAV-signal control \& system-level routing for CAVs}
\label{flow_Everything Optimized}
    \end{subfigure}

    \caption{Comparison of network performance under different control configurations.}
    \label{fig:flow_four_subfigures}
\end{figure*}

\begin{table}[t]
\centering
\renewcommand{\arraystretch}{1.15}
\setlength{\tabcolsep}{5pt}
\footnotesize
\begin{tabular}{@{}cl cccc c@{}}
\hline\hline
\textbf{Method} & \textbf{O} & \multicolumn{4}{c}{\textbf{Destination}} & \multirow{2}{*}{\textbf{Avg}} \\
\cline{3-6}
 & & \textbf{10} & \textbf{11} & \textbf{15} & \textbf{16} & \\
\hline
\multirow{4}{*}{\rotatebox{90}{\scriptsize\shortstack{Fixed sig.\\100\% HDVs}}}
 & 1  & 18.24 & 15.52 & 20.26 & 17.66 & 17.92 \\
 & 2  & 13.46 & 14.75 & 15.56 & 11.11 & 13.72 \\
 & 13 & 16.33 & 12.43 & 14.09 & 16.00 & 14.71 \\
 & 20 & 12.25 & 14.13 &  4.58 &  4.79 &  8.94 \\
\hline
\multirow{4}{*}{\rotatebox{90}{\scriptsize\shortstack{Opt.\ sig.\\100\% HDVs}}}
 & 1  & 17.39 & 12.30 & 19.80 & 15.26 & 16.19 \\
 & 2  & 11.82 & 13.45 & 14.09 &  9.06 & 12.11 \\
 & 13 & 12.67 & 11.19 & 10.53 & 15.21 & 12.40 \\
 & 20 & 11.98 & 14.14 &  4.17 &  4.77 &  8.77 \\
\hline
\multirow{4}{*}{\rotatebox{90}{\scriptsize\shortstack{Opt.\ sig.\\50\% CAVs\\no routing}}}
 & 1  & 17.53 & 14.01 & 19.62 & 15.15 & 16.58 \\
 & 2  & 12.03 & 13.56 & 14.25 &  8.99 & 12.21 \\
 & 13 & 12.54 & 12.20 & 10.47 & 15.36 & 12.64 \\
 & 20 & 11.01 & 13.28 &  4.13 &  4.81 &  8.31 \\
\hline
\multirow{4}{*}{\rotatebox{90}{\scriptsize\shortstack{Opt.\ sig.\\50\% CAVs\\CAV routed}}}
 & 1  & 15.91 & 13.99 & 19.56 & 15.41 & 16.22 \\
 & 2  & 10.75 & 11.83 & 13.06 &  9.19 & 11.21 \\
 & 13 & 12.97 & 11.54 & 10.43 & 15.01 & 12.49 \\
 & 20 &  8.25 & 10.10 &  4.18 &  4.64 &  6.79 \\
\hline\hline
\end{tabular}
\caption{Average travel time (min) per OD pair}\label{tab:od_travel_times}
\end{table}

\section{Simulation results}
To validate our proposed approach, we conducted simulations on the Sioux Falls network based on SUMO (Fig.~\ref{fig:network}) over a total simulation horizon of 35 minutes (2100 s). The desired and maximum speed on each link was set to $13.89 \text{ m/s}$ ($\approx 50 \text{ km/h}$), with acceleration bounded by $u_{\min} = -6 \text{ m/s}^2$ and $u_{\max} = 5 \text{ m/s}^2$. The length of the intersection control zone was set to 200 m. Travel demand originates at nodes 1, 2, 13, and 20, with destinations at nodes 10, 11, 15, and 16. The demand was set to 250 veh/h for each origin-destination (OD) pair. We select $K=7$ paths per OD pair for upper-level routing in k-shortest path algorithm. The upper-level travel-time measurements are updated using a 5-minute rolling window. That is, when evaluating the travel time in each link, we averaged travel time associated with the last 5 minutes. We considered four control configurations:
(i) \textit{Baseline}: 100\% HDVs with fixed signal timings and shortest-path routing; 
(ii) \textit{Signal Optimization Only}: 100\% HDVs with optimized signals (Section~\ref{sec:signal_opt_problem}) and shortest-path routing; 
(iii) \textit{Signal + CAV Control}: 50\% CAV penetration rate utilizing the lower-level joint CAV–signal optimization (Section~\ref{section_joint_cav_signals}) without upper-layer routing; and 
(iv) \textit{Proposed}: The proposed hierarchical framework, 50\% CAV penetration rate with both lower-level joint optimization and upper-level routing for CAVs.

Fig.~\ref{fig:Total vehicles in the network as a function of time} illustrates the evolution of total vehicles within the network across the four scenarios. The baseline case consistently shows the highest vehicle accumulation. Implementing the lower-level joint signal–trajectory control (50\% CAV penetration) yields a further, albeit moderate, reduction in active vehicles compared to signal optimization alone. Finally, the proposed hierarchical optimization achieves the lowest vehicle accumulation, reducing it by approximately 300 vehicles compared to the baseline case of 100\% HDVs without traffic signal optimization, thereby demonstrating improved discharge efficiency.

Fig. \ref{fig:flow_four_subfigures} illustrates the spatiotemporal traffic distribution across the network for the same scenarios. For each scenario, the heatmap depicts the average number of vehicles on each edge, while the histogram aggregates the total traffic volume accumulated per edge over the entire simulation horizon. Comparing the baseline (Fig. \ref{flow_100HDVs_nothing_optimized}) to the signal optimization-only method (Fig. \ref{flow_100HDVs_signals_optimized}), the latter noticeably reduces the intensities of traffic. Further congestion alleviation is observed in Fig.~\ref{flow_50HDVs_signals_optimized_and_CAV_optimized_no_routing} when low-level CAV trajectory optimization is activated. Finally, Fig.~\ref{flow_Everything Optimized} demonstrates the distinct advantage of the full hierarchical framework proposed in this work. By dynamically coupling the layers, the upper-level controller proactively redirects CAV flows to alternative paths (marked by the green boxes). This relieves the most heavily congested links and achieves a significantly more balanced network-wide traffic distribution. Also, this balanced distribution reveals a significant reduction in travel time; as summarized in Table~\ref{tab:od_travel_times}, the full hierarchical framework achieves the lowest average travel times across all origin-destination (OD) pairs compared to employing either control layer in isolation.


\section{Concluding Remarks}

In this paper, we proposed a hierarchical control framework for mixed traffic networks that integrates system-level routing with the joint optimization of traffic signals and CAV trajectories at signalized intersections. The proposed architecture couples a macroscopic routing layer, which assigns paths to CAVs using aggregated link-level traffic information, with a microscopic lower layer, which optimizes signal phase durations and computes efficient CAV trajectories in the vicinity of each intersection. In this way, the framework connects network-level route guidance with local intersection control while accounting for the presence of both CAVs and HDVs. The proposed approach was evaluated in the Sioux Falls network using SUMO, showing superior performance than applying either control layer in isolation. Future work should account for communication among adjacent intersections, together with dynamic rerouting of CAVs. Moreover, incorporating behavioral models that capture how HDVs respond to routing recommendations could provide further insight into their effect on overall system performance. Another promising direction is to extend the framework to multimodal systems (e.g., public transit and carpooling) and assess mobility accessibility impacts.

\linespread{0.99}\selectfont
\bibliographystyle{ieeetr}
 \bibliography{bibliography/bibliography,bibliography/Filippos,bibliography/IDS_Publications_03272026}

@ARTICLE{tzortzoglou2026toward,
  author={Tzortzoglou, Filippos N. and Malikopoulos, Andreas A.},
  journal={IEEE Potentials}, 
  title={Teaching cars to drive}, 
  year={2025},
  volume={44},
  number={6},
  pages={15-24},
  doi={10.1109/MPOT.2026.3664382}}

@INPROCEEDINGS{Salazar2019congestion,
  author={Salazar, Mauro and Tsao, Matthew and Aguiar, Izabel and Schiffer, Maximilian and Pavone, Marco},
  booktitle={2019 18th European Control Conference (ECC)}, 
  title={A Congestion-aware Routing Scheme for Autonomous Mobility-on-Demand Systems}, 
  year={2019},
  volume={},
  number={},
  pages={3040-3046},
  doi={10.23919/ECC.2019.8795897}}

@article{alessandrini2015automated,
  title={Automated vehicles and the rethinking of mobility and cities},
  author={Alessandrini, Adriano and Campagna, Andrea and Delle Site, Paolo and Filippi, Francesco and Persia, Luca},
  journal={Transportation Research Procedia},
  volume={5},
  pages={145--160},
  year={2015},
  publisher={Elsevier}
}

@article{wollenstein2021routing,
  title={Routing and rebalancing intermodal autonomous mobility-on-demand systems in mixed traffic},
  author={Wollenstein-Betech, Salom{\'o}n and Salazar, Mauro and Houshmand, Arian and Pavone, Marco and Paschalidis, Ioannis Ch and Cassandras, Christos G},
  journal={IEEE Transactions on Intelligent Transportation Systems},
  volume={23},
  number={8},
  pages={12263--12275},
  year={2021},
  publisher={IEEE}
}

@article{tajalli2021traffic,
  title={Traffic signal timing and trajectory optimization in a mixed autonomy traffic stream},
  author={Tajalli, Mehrdad and Hajbabaie, Ali},
  journal={IEEE Transactions on Intelligent Transportation Systems},
  volume={23},
  pages={6525--6538},
  year={2021},
  publisher={IEEE}
}

@article{guo2023cotv,
  title={CoTV: Cooperative control for traffic light signals and connected autonomous vehicles using deep reinforcement learning},
  author={Guo, Jiaying and Cheng, Long and Wang, Shen},
  journal={IEEE Transactions on Intelligent Transportation Systems},
  volume={24},
  number={10},
  pages={10501--10512},
  year={2023},
  publisher={IEEE}
}

@article{kamal2019development,
  title={Development and evaluation of an adaptive traffic signal control scheme under a mixed-automated traffic scenario},
  author={Kamal, Md Abdus Samad and Hayakawa, Tomohisa and Imura, Jun-ichi},
  journal={IEEE Transactions on Intelligent Transportation Systems},
  volume={21},
  number={2},
  pages={590--602},
  year={2019},
  publisher={IEEE}
}

@article{rossi2018routing,
  title={Routing autonomous vehicles in congested transportation networks: structural properties and coordination algorithms: F. Rossi et al.},
  author={Rossi, Federico and Zhang, Rick and Hindy, Yousef and Pavone, Marco},
  journal={Autonomous Robots},
  volume={42},
  number={7},
  pages={1427--1442},
  year={2018},
  publisher={Springer}
}

@article{tan2026real,
  title={Real-time cooperative scheduling for CAVs at signal-free intersections via dynamic graph analysis and tree search},
  author={Tan, Xiaojun and Ding, Yuhao and Wang, Shuai},
  journal={Robotics and Autonomous Systems},
  pages={105427},
  year={2026},
  publisher={Elsevier}
}

@article{salazar2024accessibility,
  title={On accessibility fairness in intermodal autonomous mobility-on-demand systems},
  author={Salazar, Mauro and Giraldo, Sara Betancur and Paparella, Fabio and Pedroso, Leonardo},
  journal={IFAC-PapersOnLine},
  volume={58},
  number={10},
  pages={327--333},
  year={2024},
  publisher={Elsevier}
}

@article{naderi2025lane,
  title={Lane-free signal-free intersection crossing via model predictive control},
  author={Naderi, Mehdi and Typaldos, Panagiotis and Papageorgiou, Markos},
  journal={Control Engineering Practice},
  volume={154},
  pages={106115},
  year={2025},
  publisher={Elsevier}
}

@article{liu2025integrated,
  title={An Integrated Optimization Framework for Connected and Automated Vehicles and Traffic Signals in Urban Networks},
  author={Liu, Meiqi and Li, Yalan and Liu, Xiaofei and Chen, Yang and Hao, Ruochen},
  journal={Systems},
  volume={13},
  number={4},
  pages={224},
  year={2025},
  publisher={MDPI}
}

@article{niroumand2025real,
  title={Real-time network-level traffic signal and trajectory optimization with connected automated and human-driven vehicles},
  author={Niroumand, Ramin and Kafashan, Fahim and Hajibabai, Leila and Hajbabaie, Ali},
  journal={Computer-Aided Civil and Infrastructure Engineering},
  volume={40},
  number={30},
  pages={5891--5907},
  year={2025},
  publisher={Wiley Online Library}
}

@inproceedings{chu2017dynamic,
  title={Dynamic lane reversal routing and scheduling for connected autonomous vehicles},
  author={Chu, Kai Fung and Lam, Albert YS and Li, Victor OK},
  booktitle={2017 International Smart Cities Conference (ISC2)},
  pages={1--6},
  year={2017},
  organization={IEEE}
}

@article{liang2020mobility,
  title={Mobility-aware charging scheduling for shared on-demand electric vehicle fleet using deep reinforcement learning},
  author={Liang, Yanchang and Ding, Zhaohao and Ding, Tao and Lee, Wei-Jen},
  journal={IEEE Transactions on Smart Grid},
  volume={12},
  number={2},
  pages={1380--1393},
  year={2020},
  publisher={IEEE}
}

@inproceedings{dresner2007sharing,
  title={Sharing the road: Autonomous vehicles meet human drivers},
  author={DRESNER, K},
  booktitle={Proc. of 20th International Joint Conference on Artificial Intelligence (IJCAI), 2007},
  pages={1263--1268},
  year={2007}
}

@inproceedings{tzortzoglou2025safe,
	author = {Tzortzoglou, Filippos N and Beaver, Logan E and Malikopoulos, Andreas A},
	booktitle = {64th IEEE Conference on Decision and Control (CDC)},
	pages = {2146-2151},
	title = {Safe and Efficient Coexistence of Autonomous Vehicles with Human-Driven Traffic at Signalized Intersections},
	year = {2025}}

@article{tzortzoglou2024feasibility,
	author = {Tzortzoglou, Filippos N and Beaver, Logan E and Malikopoulos, Andreas A},
	journal = {IEEE Control Systems Letters},
	pages = {2057-2062},
	title = {A feasibility analysis at signal-free intersections},
	volume = {8},
	year = {2024}}

@article{bang2024emergingequity,
	author = {Bang, Heeseung and Dave, Aditya and Tzortzoglou, Filippos and Wang, Shanting and Malikopoulos, Andreas A},
	issue = {7},
	journal = {IEEE Transactions on Intelligent Transportation Systems},
	pages = {10623-10637},
	title = {On Mobility Equity and the Promise of Emerging Transportation Systems},
	volume = {26},
	year = {2024}}

@article{bang2023optimal,
	author = {Bang, Heeseung and Malikopoulos, Andreas A},
	journal = {Automatica},
	pages = {112389},
	title = {Optimal trajectory planning meets network-level routing: Integrated control framework for emerging mobility systems},
	volume = {179},
	year = {2025}}

@article{Bang2022combined,
	author = {Bang, Heeseung and Chalaki, Behdad and Malikopoulos, Andreas A},
	doi = {10.1109/LCSYS.2022.3176594},
	journal = {IEEE Control Systems Letters},
	pages = {2749-2754},
	title = {{Combined Optimal Routing and Coordination of Connected and Automated Vehicles}},
	volume = {6},
	year = {2022}}

@article{chalaki2021CSM,
	author = {Chalaki, Behdad and Beaver, Logan E. and Mahbub, A M Ishtiaque and Bang, Heeseung and Malikopoulos, Andreas A.},
	journal = {IEEE Control Systems Magazine},
	number = {6},
	pages = {20--34},
	title = {A Research and Educational Robotic Testbed for Real-time Control of Emerging Mobility Systems: From Theory to Scaled Experiments},
	volume = {42},
	year = {2022}}

@article{Malikopoulos2020,
	author = {Malikopoulos, Andreas A and Beaver, Logan E and Chremos, Ioannis Vasileios},
	journal = {Automatica},
	number = {109469},
	title = {Optimal Time Trajectory and Coordination for Connected and Automated Vehicles},
	volume = {125},
	year = {2021}}

@inproceedings{suriyarachchi2023optimization,
  title={Optimization-based coordination and control of traffic lights and mixed traffic in multi-intersection environments},
  author={Suriyarachchi, Nilesh and Quirynen, Rien and Baras, John S and Di Cairano, Stefano},
  booktitle={2023 American Control Conference (ACC)},
  pages={3162--3168},
  year={2023},
  organization={IEEE}
}

@article{guanetti2018control,
  title={Control of connected and automated vehicles: State of the art and future challenges},
  author={Guanetti, Jacopo and Kim, Yeojun and Borrelli, Francesco},
  journal={Annual reviews in control},
  volume={45},
  pages={18--40},
  year={2018},
  publisher={Elsevier}
}

@article{xu2022general,
  title={A general framework for decentralized safe optimal control of connected and automated vehicles in multi-lane signal-free intersections},
  author={Xu, Huile and Xiao, Wei and Cassandras, Christos G and Zhang, Yi and Li, Li},
  journal={IEEE Transactions on Intelligent Transportation Systems},
  volume={23},
  number={10},
  pages={17382--17396},
  year={2022},
  publisher={IEEE}
}

@article{le2024distributed,
  title={Distributed Optimization for Traffic Light Control and Connected Automated Vehicle Coordination in Mixed-Traffic Intersections},
  author={Le, Viet-Anh and Malikopoulos, Andreas A},
  journal={IEEE Control Systems Letters},
  year={2024},
  publisher={IEEE}
}

@article{Lakshmikantham1981LargescaleDS,
  title={Large-scale dynamic systems: Stability and structure [Book reviews]},
  author={V. Lakshmikantham},
  journal={IEEE Transactions on Automatic Control},
  year={1981},
  volume={26},
  pages={976-977},
  url={https://api.semanticscholar.org/CorpusID:10253654}
}

@article{YILDIRIMOGLU2018hierarchical,
title = {Hierarchical control of heterogeneous large-scale urban road networks via path assignment and regional route guidance},
journal = {Transportation Research Part B: Methodological},
volume = {118},
pages = {106-123},
year = {2018},
issn = {0191-2615},
doi = {https://doi.org/10.1016/j.trb.2018.10.007},
url = {https://www.sciencedirect.com/science/article/pii/S0191261518301152},
author = {Mehmet Yildirimoglu and Isik Ilber Sirmatel and Nikolas Geroliminis}
}

@misc{TransportationNetworksRepo,
  author = {{Transportation Networks for Research Core Team}},
  title = {Transportation Networks for Research},
  howpublished = {\url{https://github.com/bstabler/TransportationNetworks}},
  note = {Accessed: 2024}
}

@ARTICLE{Zhu2024Coverage,
  author={Zhu, Pengbo and Sirmatel, Isik Ilber and Ferrari-Trecate, Giancarlo and Geroliminis, Nikolas},
  journal={IEEE Transactions on Control Systems Technology}, 
  title={A Coverage Control-Based Idle Vehicle Rebalancing Approach for Autonomous Mobility-on-Demand Systems}, 
  year={2024},
  volume={32},
  number={5},
  pages={1839-1853},
  doi={10.1109/TCST.2024.3375765}}

\end{document}